\documentclass[11pt]{article}
\usepackage[margin=1in]{geometry}
\usepackage[T1]{fontenc}
\usepackage{amsmath,amssymb,amsthm}
\usepackage{array,booktabs}
\usepackage{microtype}
\usepackage{tikz}
\usetikzlibrary{arrows.meta}
\usepackage[hidelinks]{hyperref}
\hypersetup{
  pdftitle={AID: A Framework for AI Infrastructure Dynamics},
  pdfauthor={Abi Aryan},
  pdfsubject={State representation and learning for coupled AI inference infrastructure},
  pdfkeywords={AI infrastructure, dynamical systems, state abstraction, partial observability, digital twins}
}
\newtheorem{proposition}{Proposition}
\title{AID: A Framework for AI Infrastructure Dynamics}
\author{Abi Aryan\\
Joule (\href{https://joule.lat}{joule.lat})\\
\small Correspondence: \href{mailto:res@abiaryan.com}{\texttt{res@abiaryan.com}}}
\date{}
\begin{document}
\maketitle

\begin{abstract}
A useful representation of AI inference infrastructure must preserve the state variables and interactions that determine future service under changing operating conditions. We introduce AID (AI Infrastructure Dynamics), a framework for representing AI infrastructure as a coupled system spanning physical, computational, networking, and serving processes. The formulation represents system state, observations, actions, disturbances, persistent configuration, and service outcomes, while allowing structured and variable-size state, asynchronous observations, multiple physical timescales, and demand that responds to service. We distinguish representations that support forecasting under an existing policy from those that preserve service-relevant behavior under changed actions, and separate both from identification of intervention responses. Two analytical results characterize a prediction-error lower bound under observational ambiguity and a sufficient condition for exact controlled state reduction. These results apply established information and state-abstraction principles to the AID formulation. We then define validation criteria for evaluating whether candidate representations preserve information relevant to cache structure, workload history, measurement availability, and imposed actions. This paper presents the AID formulation and validation criteria; it does not report empirical facility-level validation.
\end{abstract}

\section{Introduction}
A request's latency depends on both external demand and the infrastructure's current ability to serve it. That ability may depend on cache structure, request progress, network contention, resource allocation, and physical conditions that aggregate telemetry does not fully describe. Dynamic KV-memory management, iteration-level scheduling, and thermal- and power-aware serving illustrate how these layers interact~\cite{kwon2023,yu2022,stojkovic2025}.

We introduce AID (AI Infrastructure Dynamics), a framework for representing AI infrastructure as a coupled system. AID formalizes the plant through a state spanning physical, computational, networking, and serving processes, together with the observations available to an observer, the actions applied to the system, external disturbances, persistent configuration and operating constraints, and service outcomes. The state may contain structured and variable-size objects, and its evolution may occur across multiple physical timescales and operating regimes.

This representation provides a basis for asking what information a learned state must preserve. A representation may predict future behavior accurately under the policy that generated its training data while omitting information needed when routing, batching, resource allocation, power control, or workload changes. Prediction under one policy therefore does not by itself establish that the representation preserves the dynamics relevant to changed actions.

Figure~\ref{fig:aid-boundary} summarizes the system boundary and its cross-layer interactions.

\begin{figure}[tbp]
\centering
\begin{tikzpicture}[
  font=\small,
  aidbox/.style={draw, rounded corners=2pt, align=center, fill=black!3,
    text width=2.5cm, minimum height=0.75cm, inner sep=4pt},
  flow/.style={-{Latex[length=2mm]}, semithick},
  coupling/.style={{Latex[length=1.6mm]}-{Latex[length=1.6mm]}, thin}
]
\draw[rounded corners=3pt, semithick] (-3.65,-2.05) rectangle (3.65,2.05);
\node[font=\small\bfseries] at (0,1.8) {AI infrastructure: coupled state $\boldsymbol{x}_t$};
\node[aidbox] (demand) at (-5.55,0.2) {External demand\\and client rules};
\node[aidbox] (service) at (5.55,0.2) {Service output\\$\boldsymbol{r}_t$};
\node[aidbox] (requests) at (0,1.0) {Requests / serving};
\node[aidbox] (compute) at (-1.9,-0.25) {GPU compute};
\node[aidbox] (cache) at (1.9,-0.25) {Memory / cache};
\node[aidbox] (network) at (-1.9,-1.35) {Network};
\node[aidbox] (physical) at (1.9,-1.35) {Power / thermal state};
\draw[coupling] (requests.south west) -- (compute.north);
\draw[coupling] (requests.south east) -- (cache.north);
\draw[coupling] (compute) -- (cache);
\draw[coupling] (compute) -- (network);
\draw[coupling] (cache) -- (physical);
\draw[coupling] (network) -- (physical);
\draw[flow] (demand.east) -- (-3.65,0.2);
\draw[flow] (3.65,0.2) -- (service.west);
\node[align=center] (actions) at (-1.75,2.8) {Actions $\boldsymbol{a}_t$};
\node[align=center] (telemetry) at (1.75,2.8) {Telemetry $\boldsymbol{o}_t$};
\draw[flow] (actions.south) -- (-1.75,2.05);
\draw[flow] (1.75,2.05) -- (telemetry.south);
\node[aidbox, text width=4.4cm, fill=white] (clients) at (0,-2.95)
  {Closed-loop client state $\boldsymbol{c}_t$\\retries / session progression};
\draw[flow, dashed] (service.south) |- (clients.east);
\draw[flow, dashed] (clients.west) -| (demand.south);
\end{tikzpicture}
\caption{The AID system boundary. Solid arrows show demand, actions, observations, service, and schematic cross-layer coupling, not an exhaustive dependency graph or a serial execution pipeline. The dashed path represents service-dependent arrivals through client state; it is absent in open-loop replay. External intent and client rules are distinguished from realized demand.}
\label{fig:aid-boundary}
\end{figure}

The central learning question is:
\begin{quote}
\textbf{Which history-based representations preserve the distributions relevant to future service under admissible interventions?}
\end{quote}
This question connects the AID representation to partially observed sequential decision-making~\cite{kaelbling1998}, state estimation~\cite{sarkka2023}, controlled state abstraction~\cite{givan2003,li2006}, and causal identification~\cite{pearl2009}. We do not prescribe a particular predictive architecture. Instead, we specify the system representation, information boundary, policy class, and service criteria under which a candidate representation can be evaluated.

We develop the formulation in four parts:
\begin{enumerate}
\item \textbf{A coupled state-space formulation} for structured state, physical time, information availability, and demand that responds to service.
\item \textbf{A sufficiency framework} distinguishing forecasting, controlled-state reduction, and intervention identification for a specified policy class and available information.
\item \textbf{Two analytical diagnostics} characterizing observational ambiguity and exact controlled reduction, including a simple cache example.
\item \textbf{A validation protocol} connecting the formulation to controlled comparisons of state representations, observations, workloads, and imposed actions.
\end{enumerate}
Our contribution is to formulate an explicit dynamical representation for AI infrastructure and to specify criteria for evaluating whether a learned representation preserves the information relevant to prediction and changed actions. The framework does not introduce a new estimator, representation-learning algorithm, or empirically validated facility model; instead, it provides the system representation and analytical criteria needed to evaluate such models.

\section{Related Work}
\label{sec:related-work}
\subsection{Partial Observation and State Estimation}
Belief states in partially observable Markov decision processes summarize uncertainty from action and observation histories under a specified model~\cite{kaelbling1998}. Bayesian filtering and smoothing provide methods for estimating latent state under different assumptions, including linear-Gaussian models, nonlinear approximations, and particle-based methods~\cite{sarkka2023}.

We use this distinction between telemetry and latent state while also asking whether the chosen system boundary, available observations, and assumed dynamics are adequate. Computing a state estimate does not by itself show that the state can be inferred uniquely or that the dynamics are identifiable. A representation may support prediction under one policy while omitting information needed after an intervention.

\subsection{Controlled State Abstraction}
Givan et al. develop equivalence notions and model minimization for Markov decision processes using stochastic bisimulation~\cite{givan2003}. Li et al. distinguish classes of state abstraction according to the properties they preserve~\cite{li2006}. These results motivate evaluating a compressed state by the decisions it is intended to support, rather than only by reconstruction accuracy.

In Section~\ref{sec:diagnostics}, we apply this reduction principle to the joint distribution of the next state and service output. Preserving this joint distribution lets us compare future service under a common policy. The proof uses a standard induction argument, rather than introducing a new bisimulation theorem.

\subsection{Inference-Serving Mechanisms}
Inference infrastructure is not limited to autoregressive language-model serving. Clockwork builds a distributed DNN-serving system around predictable execution and request-level latency objectives~\cite{gujarati2020}. It provides an example of how execution models, resource management, and scheduling interact to produce service outcomes. Its evaluated setting differs from autoregressive LLM serving, where evolving request progress and cache state introduce additional representation questions.

Orca's iteration-level scheduling motivates request progress and evolving batch membership as serving-state variables~\cite{yu2022}. vLLM implements PagedAttention to manage KV memory in blocks, connecting serving decisions to allocation and sharing state~\cite{kwon2023}. SGLang's RadixAttention organizes cached prefixes for reuse, making prefix identity and structure relevant beyond total occupancy~\cite{zheng2024}. Sarathi-Serve uses chunked prefill to coordinate prompt processing and token generation, illustrating why execution phase and scheduling history matter to service~\cite{agrawal2024}.

\subsection{Disaggregation, Placement, and Execution Resources}
DistServe separates prefill and decoding across GPUs and chooses phase-specific resource allocation, parallelism, and placement with communication bandwidth in view~\cite{zhong2024distserve}. Splitwise likewise separates the phases across machines, using their different compute, memory, and power characteristics to design homogeneous and heterogeneous clusters~\cite{patel2024splitwise}. These mechanisms motivate a state that records which phase owns the work, where that work executes, and how state transfer interacts with available resources. Phase-dependent demand is not by itself evidence of a latent regime with distinct transition dynamics.

Mooncake treats KV data as a distributed resource, using CPU, DRAM, SSD, and network resources for a disaggregated cache alongside separate prefill and decoding clusters~\cite{qin2025mooncake}. Together, SGLang, Mooncake, and DistServe motivate three related state questions: what is cached, where it resides, and which execution phase and resource currently own the work~\cite{zheng2024,qin2025mooncake,zhong2024distserve}. A representation intended to support changed routing or placement must be tested against these distinctions, rather than assumed sufficient because it predicts aggregate occupancy.

TensorRT-LLM provides a production-oriented implementation reference for optimized inference, including in-flight batching, paged attention, and multi-GPU or multi-node execution~\cite{nvidiatensorrtllm2026}. We use this documentation to establish the implementation context, not as evidence for a mathematical guarantee or a comparative performance claim.

A recent survey by Jiang et al. organizes system-aware KV-cache optimization into temporal execution and scheduling, spatial placement and migration, and structural representation and retention~\cite{jiang2026kvcache}. This taxonomy situates the mechanisms above without replacing their primary references. AID asks a different question: which information must a representation retain to preserve future service under specified actions and observations? These systems expose complementary parts of that state/action/dynamics problem; they are not presented as incorrect abstractions or as validation of a single universal state representation.

\subsection{Physical Coupling and Feedback}
TAPAS studies thermal- and power-aware scheduling for LLM inference, while DynamoLLM considers dynamic inference-cluster configuration under performance and energy objectives~\cite{stojkovic2025,stojkovic2025dynamo}. These systems demonstrate that computational decisions can interact with physical infrastructure variables.

Work on metastable distributed-system failures identifies amplification mechanisms and persistent low-goodput overload~\cite{huang2022metastable}. We therefore include client retries and service feedback when they affect the workload within the chosen system boundary. This connection motivates a test for persistent overload; it does not show that a particular inference workload is metastable.

\subsection{Temporal Structure, Interventions, and Constraints}
Queueing analysis and empirical network-traffic studies provide antecedents for modelling demand through variability and temporal dependence rather than a stationary average alone~\cite{kingman1962,leland1994}. Causal inference distinguishes observational conditioning from imposed interventions~\cite{pearl2009}. Model predictive control provides one possible decision layer, but feasibility and stability claims require controller-specific assumptions~\cite{rawlings2022}.

We bring these distinctions together by specifying the infrastructure state, available information, disturbance process, and service criterion. Results from queueing, causal inference, or control apply only when their assumptions also hold for the infrastructure model.

\section{The AID Formulation}
AID represents AI infrastructure as a dynamical system with state, actions, external disturbances, persistent configuration, observations, and service outputs:
\begin{equation}
\label{eq:aid-system}
\mathcal{S}=(\mathcal{X},\mathcal{A},\mathcal{D},\Theta,
\mathcal{O},\mathcal{R},f,h,g).
\end{equation}
Here, $\boldsymbol{x}_t\in\mathcal{X}$ denotes the system state at time $t$, $\boldsymbol{a}_t\in\mathcal{A}$ an action applied to the system, $\boldsymbol{d}_t\in\mathcal{D}$ an external disturbance, $\boldsymbol{o}_t\in\mathcal{O}$ the telemetry available to an observer, and $\boldsymbol{r}_t\in\mathcal{R}$ the resulting service output. The parameter $\boldsymbol{\theta}\in\Theta$ denotes persistent configuration, such as hardware, software, or deployment choices that remain fixed over the period being modeled.

For systems that admit a numerical state representation, the formulation can be written in the familiar state-space form
\begin{align}
\boldsymbol{x}_{t+1}&=f(\boldsymbol{x}_t,\boldsymbol{a}_t,
\boldsymbol{d}_t;\boldsymbol{\theta})+\boldsymbol{\epsilon}_t,
\label{eq:state-transition}\\
\boldsymbol{o}_t&=h(\boldsymbol{x}_t;\boldsymbol{\theta})+\boldsymbol{\nu}_t,
\label{eq:observation}\\
\boldsymbol{r}_t&=g(\boldsymbol{x}_t,\boldsymbol{a}_t,
\boldsymbol{d}_t;\boldsymbol{\theta})+\boldsymbol{\xi}_t.
\label{eq:service-output}
\end{align}
The terms $\boldsymbol{\epsilon}_t$, $\boldsymbol{\nu}_t$, and $\boldsymbol{\xi}_t$ represent process, observation, and service-output uncertainty, respectively. A discrete operating regime $z_t$ can also be included in the state, allowing the transition law to depend on the current regime, for example through $f_{z_t}$.

This numerical form is a shorthand rather than a restriction on the state representation. In particular, AID allows $\boldsymbol{x}_t$ to contain structured or variable-size objects such as request sets, queues, cache structures, placement relations, and other state that cannot naturally be represented by a fixed-dimensional vector. External client inputs are treated separately because, depending on the system boundary, they may be exogenous disturbances or may themselves depend on the service produced by the infrastructure.

\subsection{Structured States and Transition Laws}
The state space need not be $\mathbb{R}^n$ or a smooth manifold. It can contain numerical physical variables together with variable-length request records, cache trees, block-placement relations, queue structures, and other combinatorial state. Cache allocation, prefix reuse, and distributed KV placement provide concrete motivation for retaining structure beyond numerical occupancy~\cite{kwon2023,zheng2024,qin2025mooncake}. Disaggregated execution also motivates recording phase ownership and the resources involved in transferring request state~\cite{zhong2024distserve,patel2024splitwise}.

For general measurable state spaces, we describe evolution with a transition law over valid states:
\begin{equation}
\label{eq:structured-transition}
\boldsymbol{x}_{t+1}\sim
K_{\Delta_t}(\cdot\mid\boldsymbol{x}_t,\boldsymbol{a}_t,
\boldsymbol{d}_t;\boldsymbol{\theta}),
\qquad \Delta_t=\tau_{t+1}-\tau_t,
\end{equation}
where $\tau_t$ denotes physical time.

We treat a vector embedding as a proposed reduction of the state, whose adequacy must be tested. Operations such as insertion, eviction, and request progression change structured objects and need not have the form of additive Gaussian perturbations. Physical transients likewise require a noise model appropriate to the process being described.

We assume throughout that state, action, history, and output spaces are standard Borel spaces, and that the mappings and kernels are measurable. These assumptions allow us to define regular conditional distributions. We interpret conditional laws almost everywhere with respect to the corresponding history distribution. We use additive-noise notation only when vector addition is well defined and the resulting state remains admissible.

\subsection{Clock and Information Semantics}
A multi-rate model updates components at declared physical intervals or events. Slow states can evolve between measurements; holding their estimates constant does not make the underlying physical states static. This distinction is important when fast processes influence slower variables. For example, transient computational activity may alter power and heat generation before a lower-rate thermal measurement records the resulting change. Timescale separation therefore does not by itself justify eliminating fast transients. Multi-rate numerical integration is an established computational approach~\cite{sundialsarkode2026}, but selecting an integration scheme does not establish that a coarse representation preserves service-level quantities such as latency tails.

We denote the information available at physical time $\tau_t$ by the random history record $\mathcal{I}_t$. This record contains previous actions and sensor records received by that time, together with measurement times, validity information, and aggregation windows. Here $\mathcal{I}_t$ is a random record, not a sigma-algebra.

We obtain a candidate learned state from this record through a mapping $\phi$:
\begin{equation}
\label{eq:learned-state}
\boldsymbol{s}_t=\phi(\mathcal{I}_t;\boldsymbol{\theta}).
\end{equation}
This definition limits the representation to information available when the decision is made. A trace aligned after an experiment may include measurements that the online controller had not yet received. We therefore distinguish retrospective analysis from online state estimation. When uncertainty about latent state is decision-relevant, the learned representation may instead be a belief distribution. A point estimate need not preserve the uncertainty required for future decision-making~\cite{kaelbling1998}.

Figure~\ref{fig:aid-representation} separates the latent state, available history, and candidate representation used for prediction and decision-making.

\begin{figure}[tbp]
\centering
\begin{tikzpicture}[
  font=\small,
  aidbox/.style={draw, rounded corners=2pt, align=center, fill=black!3,
    text width=3.0cm, minimum height=3.8cm, inner sep=5pt},
  flow/.style={-{Latex[length=2mm]}, semithick}
]
\node[aidbox] (state) at (0,0) {\textbf{Latent state}\\$\boldsymbol{x}_t$\\[4pt]
  Cache structure\\Request progress\\Queue state\\Thermal state};
\node[aidbox] (history) at (3.9,0) {\textbf{Available history}\\$\mathcal{I}_t$\\[4pt]
  Received telemetry\\Timestamps / windows\\Previous actions};
\node[aidbox] (representation) at (7.8,0) {\textbf{Learned state}\\$\boldsymbol{s}_t=\phi(\mathcal{I}_t;\boldsymbol{\theta})$\\[4pt]
  Candidate summary\\or belief state};
\node[aidbox] (prediction) at (11.7,0) {\textbf{Prediction / decision}\\[4pt]
  Future service\\under specified\\actions and policies};
\draw[flow] (state.east) -- (history.west);
\draw[flow] (history.east) -- (representation.west);
\draw[flow] (representation.east) -- (prediction.west);
\node[align=center, text width=3.3cm, font=\footnotesize] at (1.95,2.65)
  {Partial observations\\received by $\tau_t$};
\node[align=center, text width=3.3cm, font=\footnotesize] at (5.85,2.65)
  {Representation mapping $\phi$};
\node[align=center, text width=3.3cm, font=\footnotesize] at (9.75,2.65)
  {Evaluate information\\preservation};
\node[align=center, text width=14cm, font=\footnotesize] at (5.85,-2.55)
  {Current state is not directly available; history contains only records received by $\tau_t$.};
\end{tikzpicture}
\caption{From latent infrastructure state to a candidate learned representation. The observation arrow summarizes a history of partial measurements, not a direct observation of the current state. The mapping $\phi$ uses only available information. Prediction under changed actions requires testing what the representation preserves; the diagram does not establish sufficiency or identify intervention responses from logs.}
\label{fig:aid-representation}
\end{figure}

\subsection{Demand Boundary}
In open-loop replay, arrivals are supplied without response to service. In closed-loop operation, client state and declared external inputs may instead determine realized demand.

Let $\boldsymbol{c}_t$ denote client-side predictive state and $\boldsymbol{e}_t$ declared external inputs:
\begin{align}
\boldsymbol{d}^{\mathrm{work}}_t
&\sim\rho(\cdot\mid\boldsymbol{c}_t,\boldsymbol{e}_t),
\label{eq:workload-arrivals}\\
\boldsymbol{c}_{t+1}
&\sim K_c(\cdot\mid\boldsymbol{c}_t,\boldsymbol{r}_t,
\boldsymbol{e}_t).
\label{eq:client-transition}
\end{align}
Retries and response-dependent session progression motivate this feedback extension. Related amplification mechanisms have been studied in distributed systems~\cite{huang2022metastable}. To obtain a Markov representation, we include client memory and any other disturbance history needed to determine the next-state distribution. Whether $\boldsymbol{e}_t$ is exogenous depends on the chosen system boundary. When comparing policies, we match external intent and client rules, rather than requiring identical realized arrivals. The arrivals may differ because each policy produces a different service response, which in turn changes subsequent client behavior.

\section{Representation Sufficiency and Learning Tasks}
Given an AID representation, we define the learning task for a fixed configuration, a finite horizon $H$, an available information record, a class of admissible policies $\Pi_{\mathrm{adm}}$, and a disturbance law. Policies are non-anticipating: an action can depend on information already available, but not on future observations. For a controlled reduction, we compare policies that use only the retained representation and output history, or prescribe the same action sequence. If a policy can use information omitted from the representation, it may choose different actions. That would change the comparison as well as the future service distribution.

Let $\boldsymbol{Y}_H$ collect future observations and declared service quantities. For a policy $\pi$, let $\mathcal{L}^{\pi}(\boldsymbol{Y}_H\mid\mathcal{I}_t)$ denote the corresponding conditional law under the specified model.

An intervention-relevant representation seeks
\begin{equation}
\label{eq:representation-sufficiency}
\mathcal{L}^{\pi}(\boldsymbol{Y}_H\mid\mathcal{I}_t)
\approx
\mathcal{L}^{\pi}(\boldsymbol{Y}_H\mid\boldsymbol{s}_t)
\quad\text{for every declared }\pi\in\Pi_{\mathrm{adm}}.
\end{equation}
To make this approximation testable, we specify how the two distributions are compared, the permitted error, the prediction horizon, and the workloads and operating conditions over which the comparison applies. This requirement is stronger than prediction under a single logging policy. A representation may accurately forecast future behavior under the policy used to collect its training data while discarding information needed when routing, batching, resource allocation, or power control changes.

We distinguish three objectives according to what information the representation must preserve and what can be inferred from the data:
\begin{description}
\item[Forecast sufficiency:] preserving future quantities under a specified observational or logging policy.
\item[Controlled-state sufficiency:] preserving the future service-relevant distribution under a declared class of admissible actions and policies.
\item[Intervention identification:] establishing the response to an intervention from available information under assumptions sufficient to justify that inference.
\end{description}
These objectives are related but not interchangeable. A representation can be forecast-sufficient without being controlled-state sufficient, and a predictive model can exist without the corresponding intervention response being identifiable from observational data~\cite{pearl2009}. In particular, $\mathcal{L}^{\pi}$ denotes a law under the specified model, not automatically a causal law identified from logs.

Candidate learning tasks include history-to-belief estimation~\cite{sarkka2023}, structured-state compression evaluated against controlled abstraction criteria~\cite{givan2003,li2006}, conditional trajectory modelling, and calibrated service-risk prediction.

For models intended to predict interventions, we evaluate both workloads and actions not used in training. Random splits of a trace collected under one policy do not directly test whether the model predicts the response to a different action. We also specify the service quantities used for training and evaluation. Mean latency, tail events, rejections, failures, energy consumption, and thermal excursions describe different aspects of service. An accurate prediction of one quantity does not by itself establish accuracy for the others.

\section{Elementary Analytical Diagnostics}
\label{sec:diagnostics}
We use two analytical results to characterize what information a representation can and cannot preserve. The first gives a prediction-error lower bound when the available history cannot distinguish models with different intervention responses. It follows from the total-variation triangle inequality. The second gives a sufficient condition for reducing a Markov state while preserving future service under a common policy, following established state-abstraction principles~\cite{givan2003,li2006}. These results provide tests of the assumptions behind an infrastructure representation. We apply established principles to the AID formulation rather than claim new general theorems in control or state abstraction.

\subsection{Ambiguity Limits Common Predictions}
For probability measures $P,Q$, define
\begin{equation}
\label{eq:total-variation}
 d_{\mathrm{TV}}(P,Q)=\sup_B |P(B)-Q(B)|,
\end{equation}
where the supremum is over measurable sets.
\begin{proposition}[Prediction error under observational ambiguity]
\label{prop:ambiguity}
Suppose two admissible models induce the same history law $\mu$ for $\mathcal{I}_t$. Under the same declared future intervention, let $P_1(\cdot\mid i)$ and $P_2(\cdot\mid i)$ be regular conditional laws for a common future-output variable, and let $Q(\cdot\mid i)$ be a common prediction kernel. Define the measurable functions
\begin{align}
\gamma(i)&=d_{\mathrm{TV}}\bigl(P_1(\cdot\mid i),P_2(\cdot\mid i)\bigr),
\label{eq:model-separation}\\
\eta_j(i)&=d_{\mathrm{TV}}\bigl(P_j(\cdot\mid i),Q(\cdot\mid i)\bigr),
\quad j\in\{1,2\}.
\label{eq:prediction-error}
\end{align}
Then, for $\mu$-almost every $i$,
\begin{equation}
\label{eq:pointwise-ambiguity-bound}
\max\{\eta_1(i),\eta_2(i)\}\geq\frac{\gamma(i)}{2}.
\end{equation}
Moreover,
\begin{equation}
\label{eq:expected-ambiguity-bound}
\max_{j\in\{1,2\}}\int\eta_j(i)\,\mu(di)
\geq\frac{1}{2}\int\gamma(i)\,\mu(di).
\end{equation}
\end{proposition}
\begin{proof}
For the chosen versions of the conditional laws, the triangle inequality gives
\begin{equation}
\label{eq:ambiguity-triangle}
\gamma(i)\leq\eta_1(i)+\eta_2(i)
\end{equation}
almost everywhere. At least one term is therefore at least $\gamma(i)/2$. Integrating the same inequality gives
\begin{equation}
\label{eq:integrated-triangle}
\int\gamma\,d\mu\leq\int\eta_1\,d\mu+\int\eta_2\,d\mu,
\end{equation}
so at least one expected error is at least half the integrated separation.
\end{proof}
The interpretation is an information limit. If the available information cannot distinguish two admissible models that respond differently to the same intervention, then a predictor that assigns the same conditional distribution to that information must incur nonzero error for at least one model whenever the integrated separation is positive.

This does not imply that similar telemetry snapshots establish observational equivalence. Additional sensors, longer histories, or sufficiently informative interventions may distinguish the models. The result instead identifies a condition under which the available information is insufficient.

As a simple illustration, suppose telemetry records only cache occupancy. Two latent cache states can have identical occupancy and observed history while containing different prefixes. Under a common action requesting one of those prefixes, a binary cache-hit output can have different conditional laws in the two states. In the deterministic construction, one state produces a hit and the other a miss, giving $\gamma=1$ and a worst-case total-variation prediction error of at least $1/2$. The example is purely mathematical: it illustrates observational ambiguity rather than making an empirical claim about a particular cache implementation.

Figure~\ref{fig:aid-cache-aliasing} illustrates this construction and the bound in Proposition~\ref{prop:ambiguity}.

\begin{figure}[tbp]
\centering
\begin{tikzpicture}[
  font=\small,
  aidbox/.style={draw, rounded corners=2pt, align=center, fill=black!3,
    text width=2.9cm, minimum height=0.9cm, inner sep=5pt},
  flow/.style={-{Latex[length=2mm]}, semithick}
]
\node[aidbox] (firststate) at (0,1.2) {\textbf{Hidden state A}\\Cache:\\prefix A};
\node[aidbox] (secondstate) at (0,-1.2) {\textbf{Hidden state B}\\Cache:\\prefix B};
\node[aidbox, minimum height=2.1cm] (observation) at (3.8,0)
  {\textbf{Same history}\\Occupancy: one block\\No prefix identity};
\node[aidbox, minimum height=2.1cm] (action) at (7.6,0)
  {\textbf{Same action}\\Request prefix A};
\node[aidbox] (hit) at (11.4,1.2) {\textbf{In state A: hit}\\$P_1=\delta_1$};
\node[aidbox] (miss) at (11.4,-1.2) {\textbf{In state B: miss}\\$P_2=\delta_0$};
\draw[flow] (firststate.east) -- (observation.north west);
\draw[flow] (secondstate.east) -- (observation.south west);
\draw[flow] (observation.east) -- (action.west);
\draw[flow] (action.north east) -- (hit.west);
\draw[flow] (action.south east) -- (miss.west);
\node[align=center, text width=13.7cm, font=\footnotesize] at (5.7,-2.65)
  {For the binary hit indicator: $d_{\mathrm{TV}}(P_1,P_2)=1$, so any common predictor has worst-case error at least $1/2$.};
\end{tikzpicture}
\caption{Observational ambiguity in a deterministic cache example. The two admissible hidden states have the same occupancy-only history but different responses to the same imposed action. Here $\delta_1$ and $\delta_0$ are point masses for a binary cache-hit indicator. The action is prescribed in both states; equal telemetry does not imply equal response. This is a mathematical construction, not an empirical cache result.}
\label{fig:aid-cache-aliasing}
\end{figure}

\subsection{Exact Controlled Reduction}
Let $X_t$ be a Markov state of the augmented facility--client system at a fixed configuration, including clocks, predictive memory, and latent regimes as necessary. Here $X_t$, $A_t$, and $R_t$ denote the random state, action, and service output represented elsewhere by $\boldsymbol{x}_t$, $\boldsymbol{a}_t$, and $\boldsymbol{r}_t$. Let $\psi$ be a measurable reduction to a standard Borel space.

We describe the next reduced state and current service output through their joint one-step distribution under an imposed action. This structural kernel incorporates the specified external-input law and is defined for every admissible state--action pair, including pairs absent from the logging data. The conditional notation in the proposition refers to this structural law, rather than a response inferred automatically from logs.

\begin{proposition}[Exact controlled reduction]
\label{prop:reduction}
Suppose there exists a kernel $\bar K$ such that, for every full state $x$, admissible action $a$, and measurable sets $B,C$,
\begin{equation}
\label{eq:controlled-reduction}
\Pr\{\psi(X_{t+1})\in B,\ R_t\in C\mid X_t=x,A_t=a\}
=\bar K(B\times C\mid\psi(x),a).
\end{equation}
Then, under any common non-anticipating policy using only reduced-state and output history, started with the same policy memory, two initial states with the same reduction generate the same joint reduced-state and service-output trajectory law at every finite horizon.
\end{proposition}
\begin{proof}
The initial reduced states and policy memories agree. Given the same reduced history, the common policy selects the same action distribution. The assumed kernel then gives the same conditional distribution for the next reduced state and service output. Integrating over the common history establishes equality of the one-step extension; induction gives equality of the joint trajectory law at every finite horizon.
\end{proof}
The condition requires the joint law of the next reduced state and service output to depend on the current state only through $\psi(x)$ and the imposed action. It is closely related to established state-equivalence and abstraction principles~\cite{givan2003,li2006}; the induction proof applies these principles rather than introducing a new bisimulation theorem.

The condition is stronger than preserving next-state means or aggregate occupancy. It requires the relevant joint transition and service-output distribution to agree for every admissible state--action pair, including pairs not encountered under the logging policy. It is therefore a sufficient condition for controlled reduction, not a claim that such a reduction can be identified from observational data alone. The comparison uses the specified disturbance law and common policy memory; policies that use hidden full-state information fall outside its scope.

\section{Validation Protocol}
\label{sec:validation}
The formulation leads to a set of controlled tests of whether candidate representations preserve the information required for prediction and decision-making under declared interventions. The protocol is intended to test the assumptions and distinctions introduced by AID; no empirical results are reported here.

Candidate comparisons include:
\begin{itemize}
\item explicit cache structures versus occupancy-only and enriched summaries;
\item instantaneous arrival descriptors versus temporal histories;
\item timestamp-aware online estimates versus retrospectively aligned features;
\item point state estimates versus uncertainty-aware representations;
\item open-loop workload replay versus response-dependent clients;
\item observational forecasts versus forecasts under imposed actions.
\end{itemize}
Experiments should preserve request and infrastructure versions, reset conditions, measurement availability, and offered-request accounting. Forecast tests should be separated from imposed-action tests. A model that predicts future behavior under the logging policy has not thereby demonstrated that it predicts the response to a changed action.

Calibration should be evaluated using both coverage and sharpness across the declared horizon and operating domain. Negative results should identify which intervention, workload, or operating regime breaks a proposed reduction rather than being absorbed into a single aggregate score.

A reference environment can make the learning task reproducible without establishing universal adequacy. Evaluation across other configurations is required before claiming transfer. Regime switching and retry persistence remain hypotheses until supported by controlled evidence.

\section{Scope and Limitations}
This paper presents a formulation and validation methodology rather than a learned representation, estimator, or empirically validated facility model. The relationship between a candidate representation and the state of a real infrastructure system must be established through measurement and controlled evaluation. The exact reduction condition in Section~\ref{sec:diagnostics} is a sufficient condition, not a practical claim that real infrastructure admits an exact low-dimensional representation. Approximate reductions require an explicit discrepancy measure, prediction horizon, and error tolerance, together with analysis of how error changes over the horizon.

The ambiguity result depends on the model class and information available to the observer. Similar telemetry does not by itself establish observational equivalence, and additional measurements or interventions may resolve an apparent ambiguity. Intervention identification requires assumptions or experimental evidence sufficient to support the corresponding causal claim~\cite{pearl2009}. The framework does not by itself guarantee identification from observational logs.

Likewise, AID does not establish service-level objective (SLO) guarantees, energy optimality, or safe control. Such claims require explicit service populations, disturbance laws, constraints, and controller-specific analysis~\cite{rawlings2022}. Finally, the adequacy of any representation may depend on hardware, workloads, client behavior, and operating regime. A representation validated under one configuration should not be assumed to transfer to another without evidence.

\section{Conclusion}
AID formulates the problem of learning representations for AI infrastructure as a question of information preservation under future service and admissible actions. The framework distinguishes telemetry from latent state, forecasting from imposed-action prediction, supplied arrivals from response-dependent demand, and numerical embeddings from structured state. The central question is not whether a representation reconstructs the measurements that are easiest to log. It is whether the representation preserves the information required to predict future service under the decisions the system may actually take.

The resulting formulation provides explicit objects, assumptions, analytical diagnostics, and validation criteria for investigating that question. It does not assume that commonly exposed infrastructure metrics constitute a sufficient state, nor that predictive accuracy under an existing policy establishes intervention-relevant adequacy. AID therefore provides a basis for making representation claims in AI infrastructure explicit, testable, and falsifiable.

\bibliographystyle{plain}
\bibliography{references}

@article{kingman1962,
  author = {Kingman, J. F. C.},
  title = {On Queues in Heavy Traffic},
  journal = {Journal of the Royal Statistical Society: Series B (Methodological)},
  volume = {24},
  number = {2},
  pages = {383--392},
  year = {1962},
  doi = {10.1111/j.2517-6161.1962.tb00465.x},
  note = {\url{https://doi.org/10.1111/j.2517-6161.1962.tb00465.x}}
}

@article{leland1994,
  author = {Leland, Will E. and Taqqu, Murad S. and Willinger, Walter and Wilson, Daniel V.},
  title = {On the Self-Similar Nature of {Ethernet} Traffic (Extended Version)},
  journal = {IEEE/ACM Transactions on Networking},
  volume = {2},
  number = {1},
  pages = {1--15},
  year = {1994},
  doi = {10.1109/90.282603},
  note = {\url{https://doi.org/10.1109/90.282603}}
}

@article{kaelbling1998,
  author = {Kaelbling, Leslie Pack and Littman, Michael L. and Cassandra, Anthony R.},
  title = {Planning and Acting in Partially Observable Stochastic Domains},
  journal = {Artificial Intelligence},
  volume = {101},
  number = {1--2},
  pages = {99--134},
  year = {1998},
  doi = {10.1016/S0004-3702(98)00023-X},
  note = {\url{https://doi.org/10.1016/S0004-3702(98)00023-X}}
}

@article{pearl2009,
  author = {Pearl, Judea},
  title = {Causal Inference in Statistics: An Overview},
  journal = {Statistics Surveys},
  volume = {3},
  pages = {96--146},
  year = {2009},
  doi = {10.1214/09-SS057},
  note = {\url{https://doi.org/10.1214/09-SS057}}
}

@inproceedings{kwon2023,
  author = {Kwon, Woosuk and Li, Zhuohan and Zhuang, Siyuan and Sheng, Ying and Zheng, Lianmin and Yu, Cody Hao and Gonzalez, Joseph E. and Zhang, Hao and Stoica, Ion},
  title = {Efficient Memory Management for Large Language Model Serving with {PagedAttention}},
  booktitle = {Proceedings of the 29th Symposium on Operating Systems Principles},
  series = {SOSP '23},
  pages = {611--626},
  publisher = {Association for Computing Machinery},
  year = {2023},
  doi = {10.1145/3600006.3613165},
  note = {\url{https://doi.org/10.1145/3600006.3613165}}
}

@inproceedings{agrawal2024,
  author = {Agrawal, Amey and Kedia, Nitin and Panwar, Ashish and Mohan, Jayashree and Kwatra, Nipun and Gulavani, Bhargav and Tumanov, Alexey and Ramjee, Ramachandran},
  title = {Taming {Throughput-Latency} Tradeoff in {LLM} Inference with {Sarathi-Serve}},
  booktitle = {18th USENIX Symposium on Operating Systems Design and Implementation},
  series = {OSDI 24},
  pages = {117--134},
  publisher = {USENIX Association},
  year = {2024},
  url = {https://www.usenix.org/conference/osdi24/presentation/agrawal},
  note = {\url{https://www.usenix.org/conference/osdi24/presentation/agrawal}}
}

@inproceedings{zheng2024,
  author = {Zheng, Lianmin and Yin, Liangsheng and Xie, Zhiqiang and Sun, Chuyue and Huang, Jeff and Yu, Cody Hao and Cao, Shiyi and Kozyrakis, Christos and Stoica, Ion and Gonzalez, Joseph E. and Barrett, Clark and Sheng, Ying},
  title = {{SGLang}: Efficient Execution of Structured Language Model Programs},
  booktitle = {Advances in Neural Information Processing Systems},
  volume = {37},
  year = {2024},
  doi = {10.52202/079017-2000},
  note = {\url{https://doi.org/10.52202/079017-2000}}
}

@inproceedings{stojkovic2025,
  author = {Stojkovic, Jovan and Zhang, Chaojie and Goiri, {\'I}\~{n}igo and Choukse, Esha and Qiu, Haoran and Fonseca, Rodrigo and Torrellas, Josep and Bianchini, Ricardo},
  title = {{TAPAS}: Thermal- and Power-Aware Scheduling for {LLM} Inference in Cloud Platforms},
  booktitle = {Proceedings of the 30th ACM International Conference on Architectural Support for Programming Languages and Operating Systems},
  series = {ASPLOS '25},
  volume = {2},
  pages = {1266--1281},
  publisher = {Association for Computing Machinery},
  year = {2025},
  doi = {10.1145/3676641.3716025},
  note = {\url{https://doi.org/10.1145/3676641.3716025}}
}

@book{sarkka2023,
  author = {S{\"a}rkk{\"a}, Simo and Svensson, Lennart},
  title = {Bayesian Filtering and Smoothing},
  edition = {2},
  publisher = {Cambridge University Press},
  year = {2023},
  doi = {10.1017/9781108917407},
  note = {\url{https://doi.org/10.1017/9781108917407}}
}

@book{rawlings2022,
  author = {Rawlings, James B. and Mayne, David Q. and Diehl, Moritz M.},
  title = {Model Predictive Control: Theory, Computation, and Design},
  edition = {2},
  publisher = {Nob Hill Publishing},
  year = {2022},
  note = {Fourth electronic printing, April 2022. \href{https://sites.engineering.ucsb.edu/~jbraw/mpc/MPC-book-2nd-edition-4th-printing.pdf}{Author-hosted text}}
}

@inproceedings{yu2022,
  author = {Yu, Gyeong-In and Jeong, Joo Seong and Kim, Geon-Woo and Kim, Soojeong and Chun, Byung-Gon},
  title = {{Orca}: A Distributed Serving System for {Transformer-Based} Generative Models},
  booktitle = {16th USENIX Symposium on Operating Systems Design and Implementation},
  series = {OSDI '22},
  pages = {521--538},
  publisher = {USENIX Association},
  year = {2022},
  note = {\url{https://www.usenix.org/conference/osdi22/presentation/yu}}
}

@inproceedings{gujarati2020,
  author = {Gujarati, Arpan and Karimi, Reza and Alzayat, Safya and Hao, Wei and Kaufmann, Antoine and Vigfusson, Ymir and Mace, Jonathan},
  title = {Serving {DNNs} like {Clockwork}: Performance Predictability from the Bottom Up},
  booktitle = {14th USENIX Symposium on Operating Systems Design and Implementation},
  series = {OSDI '20},
  pages = {443--462},
  publisher = {USENIX Association},
  year = {2020},
  note = {\url{https://www.usenix.org/conference/osdi20/presentation/gujarati}}
}

@inproceedings{stojkovic2025dynamo,
  author = {Stojkovic, Jovan and Zhang, Chaojie and Goiri, {\'I}\~{n}igo and Torrellas, Josep and Choukse, Esha},
  title = {{DynamoLLM}: Designing {LLM} Inference Clusters for Performance and Energy Efficiency},
  booktitle = {2025 IEEE International Symposium on High Performance Computer Architecture},
  series = {HPCA},
  pages = {1348--1362},
  publisher = {IEEE},
  year = {2025},
  doi = {10.1109/HPCA61900.2025.00102},
  note = {\url{https://doi.org/10.1109/HPCA61900.2025.00102}}
}

@misc{sundialsarkode2026,
  author = {{SUNDIALS Project}},
  title = {{ARKODE}: Introduction},
  howpublished = {Official SUNDIALS documentation, version 7.9.0},
  note = {Accessed October 3, 2026. \href{https://sundials.readthedocs.io/en/v7.9.0/arkode/Introduction_link.html}{Official documentation}}
}

@article{givan2003,
  author = {Givan, Robert and Dean, Thomas and Greig, Matthew},
  title = {Equivalence Notions and Model Minimization in {Markov} Decision Processes},
  journal = {Artificial Intelligence},
  volume = {147},
  number = {1--2},
  pages = {163--223},
  year = {2003},
  doi = {10.1016/S0004-3702(02)00376-4},
  note = {\url{https://doi.org/10.1016/S0004-3702(02)00376-4}}
}

@inproceedings{li2006,
  author = {Li, Lihong and Walsh, Thomas J. and Littman, Michael L.},
  title = {Towards a Unified Theory of State Abstraction for {MDPs}},
  booktitle = {Proceedings of the Ninth International Symposium on Artificial Intelligence and Mathematics},
  year = {2006},
  note = {Author-hosted paper: \href{https://thomasjwalsh.net/pub/aima06Towards.pdf}{Full text}}
}

@inproceedings{huang2022metastable,
  author = {Huang, Lexiang and Magnusson, Matthew and Muralikrishna, Abishek Bangalore and Estyak, Salman and Isaacs, Rebecca and Aghayev, Abutalib and Zhu, Timothy and Charapko, Aleksey},
  title = {Metastable Failures in the Wild},
  booktitle = {16th USENIX Symposium on Operating Systems Design and Implementation},
  series = {OSDI '22},
  pages = {73--90},
  publisher = {USENIX Association},
  year = {2022},
  note = {\url{https://www.usenix.org/conference/osdi22/presentation/huang-lexiang}}
}

@inproceedings{zhong2024distserve,
  author = {Zhong, Yinmin and Liu, Shengyu and Chen, Junda and Hu, Jianbo and Zhu, Yibo and Liu, Xuanzhe and Jin, Xin and Zhang, Hao},
  title = {{DistServe}: Disaggregating Prefill and Decoding for Goodput-optimized Large Language Model Serving},
  booktitle = {18th USENIX Symposium on Operating Systems Design and Implementation},
  series = {OSDI '24},
  pages = {193--210},
  publisher = {USENIX Association},
  year = {2024},
  url = {https://www.usenix.org/conference/osdi24/presentation/zhong-yinmin},
  note = {\url{https://www.usenix.org/conference/osdi24/presentation/zhong-yinmin}}
}

@inproceedings{patel2024splitwise,
  author = {Patel, Pratyush and Choukse, Esha and Zhang, Chaojie and Shah, Aashaka and Goiri, {\'I}\~{n}igo and Maleki, Saeed and Bianchini, Ricardo},
  title = {{Splitwise}: Efficient Generative {LLM} Inference Using Phase Splitting},
  booktitle = {2024 ACM/IEEE 51st Annual International Symposium on Computer Architecture},
  series = {ISCA},
  year = {2024},
  note = {Author-hosted paper: \href{https://www.microsoft.com/en-us/research/wp-content/uploads/2023/12/Splitwise_ISCA24.pdf}{Full text}}
}

@inproceedings{qin2025mooncake,
  author = {Qin, Ruoyu and Li, Zheming and He, Weiran and Cui, Jialei and Ren, Feng and Zhang, Mingxing and Wu, Yongwei and Zheng, Weimin and Xu, Xinran},
  title = {{Mooncake}: Trading More Storage for Less Computation---A {KVCache}-centric Architecture for Serving {LLM} Chatbot},
  booktitle = {23rd USENIX Conference on File and Storage Technologies},
  series = {FAST '25},
  pages = {155--170},
  publisher = {USENIX Association},
  year = {2025},
  url = {https://www.usenix.org/conference/fast25/presentation/qin},
  note = {\url{https://www.usenix.org/conference/fast25/presentation/qin}}
}

@misc{nvidiatensorrtllm2026,
  author = {{NVIDIA}},
  title = {{TensorRT-LLM}: Overview},
  howpublished = {Official project documentation},
  year = {2026},
  note = {Accessed October 3, 2026. \href{https://nvidia.github.io/TensorRT-LLM/overview.html}{Official documentation}}
}

@inproceedings{jiang2026kvcache,
  author = {Jiang, Jiantong and Yang, Peiyu and Zhang, Rui and Liu, Feng},
  title = {Towards Efficient Large Language Model Serving: A Survey on System-Aware {KV} Cache Optimization},
  booktitle = {Findings of the Association for Computational Linguistics: ACL 2026},
  pages = {38450--38476},
  publisher = {Association for Computational Linguistics},
  year = {2026},
  doi = {10.18653/v1/2026.findings-acl.1916},
  url = {https://aclanthology.org/2026.findings-acl.1916/},
  note = {\url{https://doi.org/10.18653/v1/2026.findings-acl.1916}}
}
\end{document}